\documentclass[11pt]{article}

\usepackage[preprint]{acl}

\usepackage{times}
\usepackage{latexsym}

\usepackage[T1]{fontenc}

\usepackage[utf8]{inputenc}
\usepackage{amsfonts}
 \usepackage{amssymb} 

\usepackage{microtype}

\usepackage{inconsolata}

\usepackage[most]{tcolorbox}
\usepackage[table]{xcolor}
\usepackage{booktabs}
\usepackage{multirow}
\usepackage{amsmath}
\usepackage{enumitem}
\usepackage{xcolor}
\usepackage{xspace}
\usepackage{subcaption}
\usepackage{listings}

\usepackage[utf8]{inputenc}
\usepackage[T1]{fontenc}
\usepackage{fontawesome5}
\tcbuselibrary{listings, skins}

\usepackage{amsthm}

\newtheorem{proposition}{Proposition}

\newcommand{\hideit}[1]{}

\newcommand{\ours}{\textsc{AgenticRec}\xspace}

\title{\ours: A Recommendation-Oriented Agentic Framework with Progressive Tool-Integrated Reasoning Optimization}

\author{
  \textbf{Tianyi Li}\thanks{~~Equal contribution.} , 
  \textbf{Zixuan Wang}\footnotemark[1] , 
  \textbf{Guidong Lei}, 
  \textbf{Xiaodong Li}, 
  \textbf{Hui Li}\thanks{~~Corresponding author.} \\
  Xiamen University, Xiamen, China \\
  \texttt{\{litianyi, williamzixuan, leiguidong\}@stu.xmu.edu.cn} \\
  \texttt{\{xdli, hui\}@xmu.edu.cn}
}

\begin{document}
\maketitle
\begin{abstract}
Recommender agents built on Large Language Models offer a promising paradigm for personalized recommendation. However, existing agents typically suffer from a misalignment between their tool-integrated reasoning trajectories and recommendation feedback, limiting their ability to distinguish fine-grained user preferences. To address these challenges, we propose \ours, an agentic recommendation framework that formulates recommendation as a tool-integrated reasoning process over a recommendation-oriented tool suite. Built upon this framework, we further develop a dedicated two-stage training paradigm tailored for recommender agents. In the first stage, we introduce Recommendation-Oriented Trajectory Activation, optimize the agentic recommendation ability under implicit feedback. In the second stage, Progressive Preference Refinement further refines the agent through bidirectional preference reasoning over self-bootstrapped hard pairs, progressively sharpening preference boundaries. Theoretical analysis and extensive experiments demonstrate the effectiveness of \ours. Our code is available at https://anonymous.4open.science/r/AgenticRec-FB16.
\end{abstract}

\section{Introduction}
\label{intro}

\begin{figure}[t]
  \centering
  \includegraphics[width=\columnwidth]{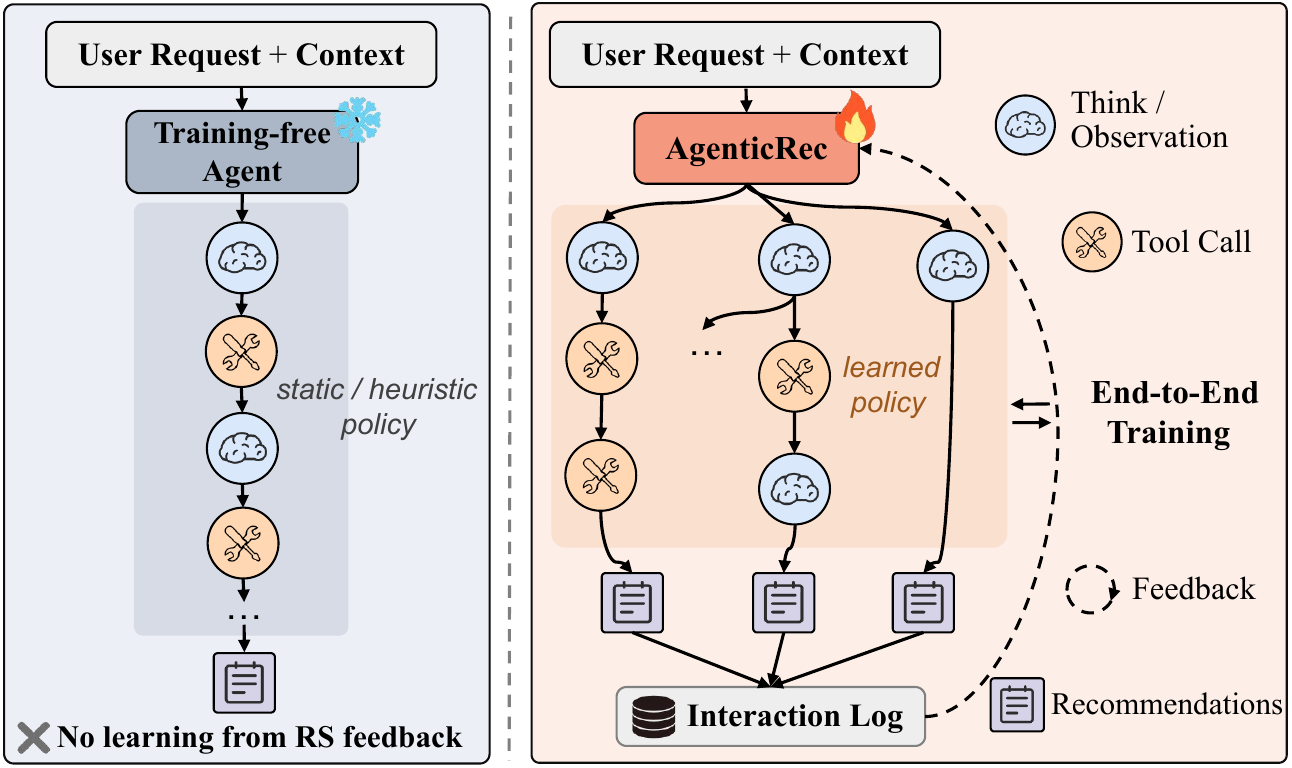}
  \vspace{-20pt}
  \caption{Compared to existing training-free recommender agents like RecMind~\cite{RecMind} and InteRecAgent~\cite{InteRecAgent} (left), \ours (right) learns recommendation-driven tool use and reasoning from real-world interaction logs.}
  \vspace{-10pt}
  \label{fig:teaser}
\end{figure}

Recommender Systems are indispensable for mitigating information overload in the digital era~\cite{ChatRecSurvey, LLM4RecSurvey}.
Recently, Large Language Models (LLMs) have significantly affected recommendation~\cite{LLM4RecSurvey,Survey4RecAgent} due to their strong natural language understanding and reasoning capabilities~\cite{deepseekr1,gpt5}, enabling semantic matching between complex user intents and diverse item descriptions.
However, recommendation is ultimately based on large-scale implicit feedback and collaborative signals, which cannot be fully captured by language priors alone~\cite{LLM4RecSurvey,LLM4RecSurvey2}.

Early attempts adapt LLMs to recommendation via fine-tuning, typically formulating the task as next-token prediction through implicit-feedback supervision~\cite{P5, TALLRec, llara}.
Although effective, they remain largely one-shot predictors that output recommendation results from the given prompt.
Moreover, real-world systems often involve incomplete and heterogeneous information: ambiguous user intent~\cite{userintent}, missing item information~\cite{RecSurvey2}, and collaborative signals not fully accessible from text alone~\cite{llara}.

To overcome these limitations, recommender agents have recently emerged as a promising paradigm, serving as the central hub of recommendation systems~\cite{InteRecAgent} or being applied in conversational recommendation scenarios~\cite{MACRS}, leveraging agents’ strong planning, reasoning, and memory capabilities to enable more interactive, context-aware, and proactive recommendation~\cite{Survey4RecAgent,RecAgentSurvey}.

Along this direction, recent work has explored recommender agents~\cite{abs-2503-05659,Survey4RecAgent} that equip LLMs with structured decision loops (e.g., ReAct~\cite{react}) and tools to perform \emph{``Think and Act before Recommendation''} for evidence-aware recommendation~\cite{RecMind, InteRecAgent, ToolRec}.
Despite their promising progress, existing recommender agents still face two challenges: 
\begin{itemize}[leftmargin=10pt,topsep=5pt,itemsep=0.2pt]
    \item \textbf{Misalignment between Tool-Integrated Reasoning Trajectories and Recommendation Feedback:} A common line of existing methods are driven by generic language priors or handcrafted prompt heuristics~\cite{RecMind,InteRecAgent} (the left part of Fig.~\ref{fig:teaser}), leaving their reasoning trajectories not optimized under recommendation feedback. Since tool use is inherently intertwined with reasoning, not only \emph{when/what} to invoke but also \emph{how} the retrieved evidence is integrated into the final decision must be determined via reasoning, posing a challenge to the design. 

   \item \textbf{Insufficiency in Resolving Fine-Grained User Preferences:} Existing methods typically lack a mechanism to progressively sharpen fine-grained preference boundaries on hard, highly confusable candidates. In practice, implicit feedback provides sparse supervision, and inaccurate recommendations may arise from subtle preference differences among top-ranked candidates, which are difficult to correct using coarse list-level signals alone. As a result, agents may learn globally reasonable behaviors, yet remain brittle in challenging cases.
\end{itemize}

Motivated by these challenges, we propose \ours, a framework that formulates agentic recommendation and optimizes it under recommendation feedback via a dedicated two-stage training paradigm tailored to it.
To support evidence-grounded recommendation decision making, we design a suite of recommendation-oriented tools, including user profile analysis, item information retrieval, behavioral statistics, and collaborative information access, which are specifically tailored to the agent's tool-integrated reasoning recommendation workflow.
Built upon this framework, we develop a two-stage optimization paradigm. 
In the first stage, we introduce Recommendation-Oriented Trajectory Activation, which enables end-to-end optimization of tool-integrated reasoning trajectories under implicit recommendation feedback. Specifically, we optimize the entire decision trajectory through tool-integrated policy optimization.
In the second stage, we further propose Progressive Preference Refinement (PPR), a self-bootstrapped refinement strategy that progressively sharpens fine-grained preference boundaries among highly confusable candidates. Specifically, PPR mines hard competitors from the agent’s own ranking violations and performs bidirectional preference reasoning to transform challenging recommendation errors into refinement supervision signals.

Our contributions are summarized as follows:
\begin{itemize}[leftmargin=10pt,topsep=3pt,itemsep=0.2pt]
    \item We propose an agentic recommendation framework that integrates a suite of carefully designed recommendation-oriented tools into the reasoning process to support evidence-grounded recommendation, together with a tool-integrated policy optimization for training recommender agents.

    \item Within the proposed framework, we develop a two-stage training paradigm consisting of Recommendation-Oriented Trajectory Activation for end-to-end optimization of tool-integrated reasoning trajectories, and Progressive Preference Refinement for progressively sharpening fine-grained preference discrimination among highly confusable items.

    \item  Theoretical analysis and extensive experiments on multiple recommendation benchmarks demonstrate the recommendation effectiveness of \ours.
\end{itemize}

\section{Related Work}
\label{related_work}

\begin{figure*}[t]
  \centering
  \includegraphics[width=0.97\textwidth]{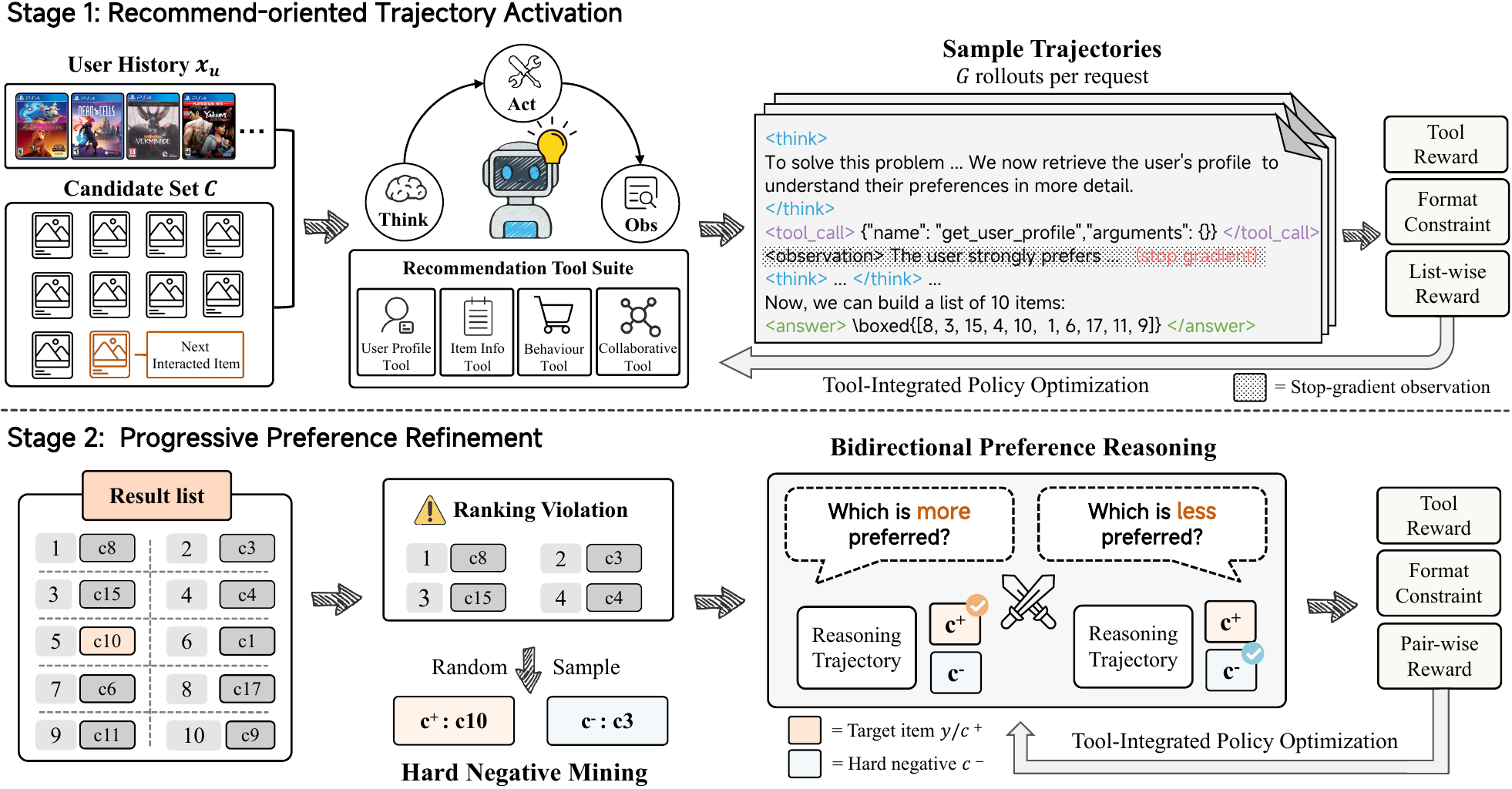}
  \caption{Overview of \ours.}
  \label{fig:model}
\end{figure*}

\paragraph{LLMs for Recommendation.} 
LLMs have been explored for recommendation through prompting~\cite{ChatRec, LLMRank} and task-unified fine-tuning~\cite{P5,TALLRec,CoLLM,llara}. Recent studies further introduce reasoning to enhance multi-step recommendation decisions~\cite{abs-2503-24235, Reason4Rec, Rec-R1, R2EC, OneRec-Think}. 

\paragraph{Recommender Agents.} 
Beyond treating LLMs as black-box predictors, existing studies have introduced the LLM-agent paradigm into recommendation, which can be categorized into three directions~\cite{Survey4RecAgent}: \textit{simulation-oriented agents}~\cite{Agent4Rec, AgentCF, UserSimulator, RecAgent} that emulate user behaviors for analysis and data generation, \textit{interaction-oriented agents}~\cite{MACRS, Hybrid-MACRS, CRAVE} that handle users' natural-language requests, and \textit{recommender-oriented agents}~\cite{RecMind, InteRecAgent, STARec, Mr.Rec, MGFRec} that directly leverage the general intelligence and tool-use ability of LLMs to provide recommendations, which is the focus of this paper. 
Recommender agents can not only improve recommendation quality and interpretability, but also serve as AI assistants in real-world applications, showing great potential to reshape users' online shopping experiences.
\section{Problem Formulation}
\label{sec:problem_definition}
Following prior recommender agent setting~\cite{RecMind,iAgent}, given a user interaction history $x_u=[v_1,v_2,\dots,v_t]$ and a candidate set $\mathcal{C}=\{c_1,c_2,\dots,c_n\}$, the recommender agent aims to generate a ranked list
$r_K=[c_{\sigma(1)},c_{\sigma(2)},\dots,c_{\sigma(K)}]$
of the Top-$K$ items that are most likely to match the user’s preferences, where $\sigma(\cdot)$ denotes the ranking order over candidate items and $K \le n$.

Our goal is to learn an agent policy $\pi_\theta$ that maximizes the expected list-wise utility under real-world feedback $y$:
\[
\max_{\theta}\ \mathbb{E}_{(x_u,\mathcal{C},y)\sim\mathcal{D}}\big[R(r_K, y)\big], 
\]
where $R(\cdot)$ measures the ranking quality induced by implicit feedback (e.g., NDCG@K).
\section{\ours} 
\label{sec:method}

In this section we introduce \ours, an agentic recommendation framework which couples a tool-integrated reasoning  workflow with a recommendation-oriented tool suite (Sec.~\ref{sec:framework}).
Built upon this framework, we develop a two-stage training paradigm:
(1) Recommendation-oriented Trajectory Activation (Sec.~\ref{sec:listwise_grpo}) end-to-end activates the agentic recommendation ability;
(2) Progressive Preference Refinement (Sec.~\ref{sec:ppr}) further sharpens fine-grained preference discrimination through bidirectional reasoning over self-bootstrapped hard negatives.
Fig.~\ref{fig:model} provides an overview of \ours.

\subsection{Overall Framework}
\label{sec:framework}
To support evidence-grounded recommendation decision making, we formulate \ours as a Tool-Integrated Reasoning agentic recommendation workflow that interleaves deep reasoning with designed recommendation-oriented tool suite. This enables the agent to go beyond the fixed linguistic knowledge of LLMs and autonomously explore across different task scenarios.

\paragraph{Tool-Integrated Reasoning Recommendation}
Following the ReAct paradigm~\cite{react}, the workflow forms a closed loop that interleaves reasoning and tool invocation for recommendation.
The agent iteratively generates \textit{Think} tokens, optionally executes an \textit{Act} by selecting a tool, receives its \textit{Observation (Obs)}, and finally outputs a \textit{Recommendation (Rec)} list:
\[
\textit{Think}_1 \rightarrow \textit{Act}_1 \rightarrow \textit{Obs}_1 \rightarrow \cdots \rightarrow \textit{Rec}.
\]

\paragraph{Recommendation-Oriented Tool Suite}
Drawing on practical experience from recommender-system applications, we define a recommendation-specific tool suite. Notably, our framework allows additional tools to be seamlessly integrated. We briefly introduce the main tools below, with more details provided in Appendix~\ref{app:tool}.

\paragraph{\normalfont\itshape (1) User Profile Tool.}
It retrieves an LLM-pre-generated profile summarizing user information and long-term preferences, such as favored categories and interests inferred from historical interactions. It can also be extended to update profiles with newly observed feedback.

\paragraph{\normalfont\itshape (2) Item Information Tool.}
It provides query interfaces for item attributes, such as descriptions and prices, helping the agent understand candidate items on demand during reasoning.

\paragraph{\normalfont\itshape (3) Behavioral Statistics Tool.}
It extracts features helpful for recommendation based on user behavior data, turning manual feature engineering into an agentic process to better utilize vast behavioral records.

\paragraph{\normalfont\itshape (4) Collaborative Information Tool.}
It maintains user and item embeddings containing collaborative signals, which support user-to-user and item-to-item retrieval, enabling the agent to reference similar user groups and related items for recommendation. Such collaborative reasoning is difficult to infer from LLM priors alone.

\subsection{Recommendation-Oriented Trajectory Activation}
\label{sec:listwise_grpo}

As the first training stage, Recommendation-oriented Trajectory Activation (RTA) aims to activate the agentic recommendation ability via tool-integrated policy optimization, where recommendation-specific rewards are propagated through the full reasoning and tool-invocation trajectory through reinforcement learning.

\paragraph{Reward Design} 
We design the reward function by considering three aspects: recommendation quality, output format, and tool-use control.
Formally, the reward is defined as:
\begin{equation}
\label{eq:main_reward}
\resizebox{\linewidth}{!}{$
    R(r_K, y)=
    \begin{cases}
        R_{\text{rec}}(r_K, y) & \text{if } \text{valid}(r_K)\ \text{and}\ y \in r_K, \\
        -0.5 & \text{if } \text{valid}(r_K)\ \text{and}\ y \notin r_K, \\
        -1 & \text{otherwise},
    \end{cases}
$}
\end{equation} 
where $R_{\text{rec}}$ denotes the recommendation reward, and $\text{valid}(r_K)$ indicates that the output satisfies all validity constraints. The detailed designs are described as follows.

\paragraph{\normalfont\itshape (1) Recommendation Reward.}
We model $R_{\text{rec}}$ by computing the NDCG score of the user's ground-truth interacted item in the recommendation list, so as to align the ranking objective with user preference signals. Formally,
\begin{equation}
\label{eq:rankreward}
    R_{\text{rec}}(r_K, y)=\mathrm{NDCG@}K(r_K; y).
\end{equation}
If $y \notin r_K$, the recommendation is considered unsuccessful and receives the penalty $-0.5$ in Eq.~\ref{eq:main_reward}.

\paragraph{\normalfont\itshape (2) Format Constraint.}
Invalid outputs (e.g., invalid items or tools) lead to execution failures in real-world applications.
We therefore assign a penalty of $-1$ when the format is incorrect, enforcing strict validity of generated recommendation.

\paragraph{\normalfont\itshape (3) Tool-call Bonus.}
To encourage tool use while avoiding reward hacking, we add a small bonus only for high-quality tool-use trajectories. After computing $R_{\text{rec}}$, we add $+0.1$ if the agent achieves Hit@1 and $N_{\text{tool}}(\tau)>0$. Here, $N_{\text{tool}}(\tau)$ denotes the number of tool invocations in $\tau$.

\paragraph{\normalfont\itshape (4) Tool-call Budget.}
To control latency and avoid over-reliance on external evidence, we limit each trajectory to at most $10$ tool interactions. Exceeding this budget makes the trajectory invalid and receives the invalid-format penalty.

\paragraph{Policy Optimization Algorithm} 
For each training instance $(x_u, \mathcal{C}, y)\sim\mathcal{D}$, we employ GRPO~\cite{deepseekmath} as our tool-integrated policy optimization algorithm with rewards in Eq.~\ref{eq:main_reward}, which estimates advantage using a group $G$ of rollouts.

\paragraph{Dynamic Sampling}
Since early training often yields uniformly low rewards and thus noisy updates, inspired by the dynamic sampling strategy in DAPO~\cite{DAPO}, we drop training instances whose rollout group contains no informative trajectories, i.e., all sampled trajectories receive negative rewards.

\subsection{Progressive Preference Refinement}
\label{sec:ppr}

After the RTA stage, a globally effective policy for recommendation has been activated in the agent.
To further strengthen the agent's capability, we introduce Progressive Preference Refinement (PPR) as a self-bootstrapped refinement stage. 

We mine hard negatives from items incorrectly ranked above the ground-truth item, and then train the agent with bidirectional preference reasoning on the resulting hard pairs.
This bidirectional thinking asks the agent not only to identify what the user is likely to prefer, but also to explicitly reason about what the user is less likely to choose. 
In this way, PPR moves beyond merely fitting positive implicit feedback, and progressively sharpens fine-grained preference boundaries.

\paragraph{Hard Negative Mining from Ranking Violations} 
Inspired by prior work~\cite{negative_mining}, we exploit the agent's own outputs to automatically mine hard negatives. For each input $(x_u, \mathcal{C})$, we identify the ground-truth positive item $c^+=y \in \mathcal{C}$ based on user feedback. After the agent generates a ranked list $r_K$, we check whether $c^+$ is placed at the top position. If $c^+$ is not ranked first, we treat this case as a \emph{ranking violation}, since at least one competing item is ranked above the positive, revealing an informative preference boundary.

Concretely, let $\operatorname{rank}_{r_K}(c)$ denote the position of item $c$ in the list $r_K$ (top-$K$). We construct a hard-negative candidate set by collecting items within the top-$K$ that are ranked above the positive; if the positive item is not in $r_K$, we use the entire top-$K$ list as the hard-negative pool:
\begin{equation}
    \resizebox{\linewidth}{!}{$
    H(x_u, \mathcal{C}) = 
        \begin{cases}
            \{  c^-\mid c^- \in r_K,\, \operatorname{rank}_{r_K}(c^-) < \operatorname{rank}_{r_K}(c^+) \}, & \text{if } c^+ \in r_K, \\
            r_K, & \text{if } c^+ \notin r_K .
        \end{cases}
    $}
\end{equation}

Instead of deterministically picking the top-ranked item, we randomly sample one negative item from $H$ to form a hard preference pair $(c^+, c^-)$. This stochastic mining increases negative diversity and avoids overly-adversarial pairs that can yield ambiguous supervision, providing a more stable refinement signal while remaining focused on competitive negatives. The mining procedure requires no external labeling beyond the implicit feedback used to identify $c^+$.

\paragraph{Bidirectional Preference Reasoning} 
The mined hard pairs expose challenging cases where the current agent fails to distinguish the true preference from strong competitors. 
The agent not only needs to identify what the user is likely to click, but also to explicitly recognize what the user is less likely to choose under the same context. 
However, training only on the positive direction provides one-sided supervision: the hard negative is used as a contrastive counterpart, while the reason why it should be rejected remains under-explored. 
To overcome the limitations of one-sided supervision, we propose \emph{bidirectional preference reasoning}, a training mechanism that turns each hard pair into two complementary tasks, thereby enhancing the agent's versatility in recommendation.

Given a mined hard preference pair $(c^+, c^-)$ under the same user context $x_u$, we refine the agent via bidirectional preference reasoning from two symmetric directions:
\begin{itemize}[leftmargin=10pt,topsep=5pt,itemsep=0.2pt]
    \item \textbf{Positive Direction:} We first task the agent with the standard objective: identifying the item that the user is more likely to interact with next. This reinforces the agent's ability to recognize matching attributes between the user behavior and the ground-truth item $c^+$.

    \item \textbf{Negative Direction:} We introduce a counter-perspective task where the agent must explicitly identify the item the user is less likely to choose. This forces the agent to actively deliberate on the specific shortcomings or mismatching features of the hard negative item $c^{-}$ relative to the user's history, rather than simply treating it as a background non-positive sample.
\end{itemize}
A prompt example of the PPR task is provided in Appendix~\ref{app:prompt}. Both directions share the same structured output format but differ in the queried perspective, thereby providing complementary learning signals for enhancing the agent's capability.

\paragraph{Training Procedure}
PPR follows the same tool-integrated policy optimization principle as RTA.
For each mined pair $(c^+,c^-)$, we construct two preference tasks corresponding to the positive and negative directions, and perform rollouts under the \ours framework. 
Different from Eq.~\ref{eq:main_reward}, PPR replaces $R_{\text{rec}}$ with a binary pair-wise reward for each direction:
\begin{equation}
    R_{\text{pair}}^{d}=\mathbb{I}[\hat{c}^{d}=c^{d}_{\star}], \quad d\in\{+,-\},
\end{equation}
where $c^{+}_{\star}=c^+$ for the positive direction, $c^{-}_{\star}=c^-$ for the negative direction, and $\hat{c}^{d}$ is the agent's predicted item under direction $d$. 
Other reward components are kept unchanged to ensure valid outputs and controlled tool invocation during preference refinement.

\subsection{Theoretical Justifications}
We provide theoretical justifications showing that both training stages of \ours are well-grounded in improving recommendation performance. Specifically,
(1) the RTA gradient estimator with the inclusive group-average baseline yields a directionally unbiased estimate of the expected list-wise utility objective (i.e., NDCG@K), ensuring accurate credit assignment over tool-integrated trajectories;
(2) the PPR objective on mined hard pairs minimizes a convex upper bound of the pair-wise ranking error probability, thereby contributing to improved ranking performance.
Together, these results indicate that the two stages of \ours are theoretically sound and jointly contribute to the final result. Formal statements and proofs are deferred to Appendix~\ref{app:theoretical_analysis}.

\section{Experiments}
\label{sec:experiments}

\subsection{Experimental Setup}
\label{sec:experimental_setup}

\begin{table*}[t]
  \centering
  \caption{Evaluation results among baselines and \ours. The best performance score is denoted in \textbf{bold}.}
  \label{tab:overall}
  \renewcommand{\arraystretch}{1.2}
  \resizebox{\textwidth}{!}{ 
    \setlength{\tabcolsep}{0.8mm}{ 
      \begin{tabular}{@{}lcccccccccccccccccccc@{}}
        \toprule
        \multicolumn{1}{c}{\multirow{2.5}{*}{\textbf{Model}}} & \multicolumn{5}{c}{\textbf{CDs}} & \multicolumn{5}{c}{\textbf{Instruments}} & \multicolumn{5}{c}{\textbf{Office}} & \multicolumn{5}{c}{\textbf{Games}} \\
        \cmidrule(l){2-6} \cmidrule(l){7-11} \cmidrule(l){12-16} \cmidrule(l){17-21}
        & H@1 & H@5 & N@5 & H@10 & N@10 & H@1 & H@5 & N@5 & H@10 & N@10 & H@1 & H@5 & N@5 & H@10 & N@10 & H@1 & H@5 & N@5 & H@10 & N@10 \\
        \midrule 
        \rowcolor{gray!8} 
        \multicolumn{21}{@{}c}{\textit{\textbf{Conventional Sequential Recommenders}}} \\
        \midrule
        Caser & 0.0488 & 0.2526 & 0.1487 & 0.5089 & 0.2307 & 0.1825 & 0.4644 & 0.3275 & 0.6384 & 0.3828 & 0.1771 & 0.4374 & 0.3139 & 0.5541 & 0.3511 & 0.1535 & 0.3953 & 0.2813 & 0.5220 & 0.3217 \\
        GRU4Rec & 0.0512 & 0.2348 & 0.1403 & 0.4851 & 0.2202 & 0.1960 & 0.4534 & 0.3291 & 0.6446 & 0.3907 & 0.1642 & 0.4324 & 0.3054 & 0.5705 & 0.3490 & 0.1708 & 0.3781 & 0.2817 & 0.5239 & 0.3288 \\
        SASRec & 0.0536 & 0.2348 & 0.1413 & 0.4660 & 0.2145 & 0.1813 & 0.4546 & 0.3238 & 0.6323 & 0.3807 & 0.1662 & 0.4270 & 0.3018 & 0.5863 & 0.3525 & 0.1535 & 0.3915 & 0.2762 & 0.5374 & 0.3234 \\
        ReaRec & 0.0357 & 0.1871 & 0.1082 & 0.4576 & 0.1942 & 0.1507 & 0.3786 & 0.2633 & 0.6188 & 0.3400 & 0.1716 & 0.4265 & 0.3049 & 0.5551 & 0.3457 & 0.1459 & 0.3723 & 0.2298 & 0.5873 & 0.2987 \\
        \midrule
        \rowcolor{gray!8}
        \multicolumn{21}{@{}c}{\textit{\textbf{Training-free LLM-based Recommenders}}} \\
        \midrule
        LLMRank & 0.0501 & 0.2265 & 0.1373 & 0.5995 & 0.2569 & 0.0417 & 0.2745 & 0.1533 & 0.6017 & 0.2584 & 0.0713 & 0.3102 & 0.1872 & 0.5576 & 0.2665 & 0.0403 & 0.2745 & 0.1519 & 0.6660 & 0.2767 \\
        InteRecAgent & 0.0782 & 0.2547 & 0.1660 & 0.4857 & 0.2392 & 0.2063 & 0.4735 & 0.3433 & 0.6515 & 0.4004 & 0.1911 & 0.4467 & 0.3211 & 0.6056 & 0.3717 & 0.1782 & 0.4107 & 0.3009 & 0.5599 & 0.3477 \\
        \midrule
        \rowcolor{gray!8} 
        \multicolumn{21}{@{}c}{\textit{\textbf{Trainable LLM-based Recommenders}}} \\
        \midrule
        TALLRec & 0.1144 & 0.3539 & 0.2313 & 0.6007 & 0.3102 & 0.0649 & 0.2683 & 0.1618 & 0.5147 & 0.2406 & 0.1539 & 0.3592 & 0.2543 & 0.5522 & 0.3318 & 0.1362 & 0.3819 & 0.2608 & 0.6468 & 0.3474 \\
        LLaRA & 0.2324 & 0.5172 & 0.3736 & 0.7234 & 0.4394 & 0.2512 & 0.5404 & 0.4043 & 0.6833 & 0.4510 & 0.2416 & 0.5007 & 0.3928 & 0.6496 & 0.4398 & 0.2763 & 0.6199 & 0.4588 & 0.7332 & 0.4952 \\
        S-DPO & 0.1764 & 0.4183 & 0.2988 & 0.6293 & 0.3669 & 0.0723 & 0.2671 & 0.1657 & 0.5220 & 0.2475 & 0.1573 & 0.3528 & 0.2513 & 0.5421 & 0.3522 & 0.2006 & 0.4529 & 0.3292 & 0.6967 & 0.4079 \\
        ReRe & 0.2073 & 0.4648 & 0.3371 & 0.6901 & 0.4080 & 0.2584 & 0.5049 & 0.3900 & 0.6765 & 0.4447 & \textbf{0.2601} & 0.5052 & 0.3970 & 0.6962 & 0.4579 & 0.2921 & 0.5528 & 0.4473 & 0.7102 & 0.4976 \\
        \midrule
        \textbf{\ours} & \textbf{0.2992} & \textbf{0.6472} & \textbf{0.4795} & \textbf{0.8093} & \textbf{0.5324} & \textbf{0.2586} & \textbf{0.6115} & \textbf{0.4393} & \textbf{0.8052} & \textbf{0.5021} & 0.2494 & \textbf{0.5755} & \textbf{0.4123} & \textbf{0.7773} & \textbf{0.4775} & \textbf{0.3282} & \textbf{0.6468} & \textbf{0.4887} & \textbf{0.8157} & \textbf{0.5445} \\
        \bottomrule
      \end{tabular}
    }
  }
\end{table*}

\begin{table}[t]
  \centering
  \small
  \caption{Ablation study of variants of \ours. R: pure \underline{r}easoning by LLM, TIRR: \underline{t}ool-\underline{i}ntegrated \underline{r}easoning \underline{r}ecommendation.}
  \label{tab:ablation}
  \renewcommand{\arraystretch}{0.8} 
  \setlength{\tabcolsep}{5pt}
  \resizebox{\linewidth}{!}{
  \begin{tabular}{lllccc}
    \toprule
    \textbf{Data} & \textbf{Setting} & \textbf{Variant} & \textbf{H@1} & \textbf{H@5} & \textbf{N@5} \\
    \midrule
    \multirow{5}{*}{CDs} & 
    \multirow{2}{*}{Frozen} & R & 0.1907 & 0.4696 & 0.3303 \\
    & & TIRR & 0.2331 & 0.5476 & 0.3955 \\
    \cmidrule(lr){2-6} & 
    \multirow{3}{*}{Trained} & R & 0.2324 & 0.5399 & 0.3865 \\
    & & TIRR (RTA) & 0.2837 & 0.6162 & 0.4606 \\
    & & TIRR (RTA+PPR) & \textbf{0.2992} & \textbf{0.6472} & \textbf{0.4795} \\
    \midrule
    \multirow{5}{*}{Instr.} 
    & \multirow{2}{*}{Frozen} & R & 0.1973 & 0.5613 & 0.3818 \\
    & & TIRR & 0.1948 & 0.5277 & 0.3673 \\
    \cmidrule(lr){2-6}
    & \multirow{3}{*}{Trained} & R & 0.2328 & 0.5797 & 0.4141 \\
    & & TIRR (RTA) & 0.2463 & 0.6091 & 0.4321 \\
    & & TIRR (RTA+PPR) & \textbf{0.2586} & \textbf{0.6115} & \textbf{0.4393} \\
    \midrule
    \multirow{5}{*}{Office} 
    & \multirow{2}{*}{Frozen} & R & 0.2147 & 0.5017 & 0.3640 \\
    & & TIRR & 0.1963 & 0.4913 & 0.3471 \\
    \cmidrule(lr){2-6}
    & \multirow{3}{*}{Trained} & R & 0.2276 & 0.5374 & 0.4004 \\
    & & TIRR (RTA) & 0.2326 & 0.5572 & 0.4040 \\
    & & TIRR (RTA+PPR) & \textbf{0.2494} & \textbf{0.5755} & \textbf{0.4123} \\
    \bottomrule
  \end{tabular}
  }
\end{table}

\paragraph{Task Setting}
We design the task setting according to the problem definition in Sec.~\ref{sec:problem_definition}. 
Specifically, following common practice in prior works~\cite{llara,S-DPO}, each sample contains a candidate set of 20 items, including one ground-truth next item as the positive item and 19 negative items randomly sampled from the item pool. 
Given the textual input of user behavior records, the model is required to generate a top-10 ranked list from the candidate set.

\paragraph{Datasets and Metrics}
All training and evaluation samples are constructed from four subsets of the Amazon Reviews 2023\footnote{\url{https://amazon-reviews-2023.github.io/}}: CDs, Instruments, Office, and Games. Recommendation quality is evaluated by NDCG@K (N@K) and HitRate@K (H@K) of the user's actual purchased item in the generated list, where $K \in \{1, 5, 10\}$.

\paragraph{Backbone LLM and Baselines}
We use Qwen3-4B-Instruct-2507~\cite{Qwen3} as \ours's backbone LLM. We compare \ours with baselines covering conventional sequential recommenders, training-free LLM-based recommenders, and trainable LLM-based recommenders, to evaluate its performance.

\par\addvspace{0.5em}\noindent
Full experimental details are in Appendix~\ref{app:experimental_details}.

\subsection{Main Results}
\label{sec:main_results}


Table~\ref{tab:overall} reports the overall comparison on four Amazon subsets. Overall, \ours achieves the best performance on most datasets and metrics, obtaining the leading results in 19 out of 20 columns. The consistent improvements on both H@K and N@K show that \ours improves both item matching and ranking quality.

The comparison among different groups shows that the world knowledge and reasoning ability of LLMs can be useful for recommendation, that post-training with recommendation data can further improve the effectiveness of LLM-based recommenders, and that the agent paradigm is also beneficial. By combining these advantages in a tool-integrated agent framework and optimizing it with reinforcement learning, \ours achieves stronger overall performance, highlighting the potential of trained recommender agents.

\subsection{Ablation Study}
\label{sec:ablation_study}

\begin{figure}[t]
    \centering
    \includegraphics[width=1\linewidth]{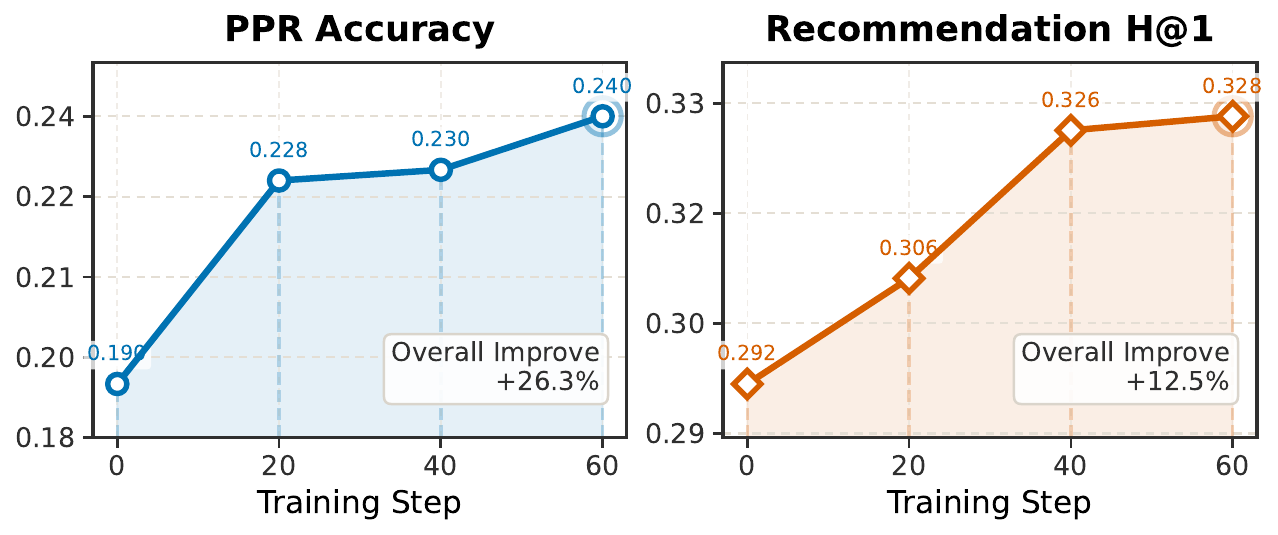}
    \caption{Performance changes during PPR training on the Games dataset, including the agent's preference judgment accuracy on the PPR test set and H@1 on the final recommendation test set.}
    \label{fig:ppr_training}
\end{figure}

\begin{figure}
    \centering
    \includegraphics[width=1\linewidth]{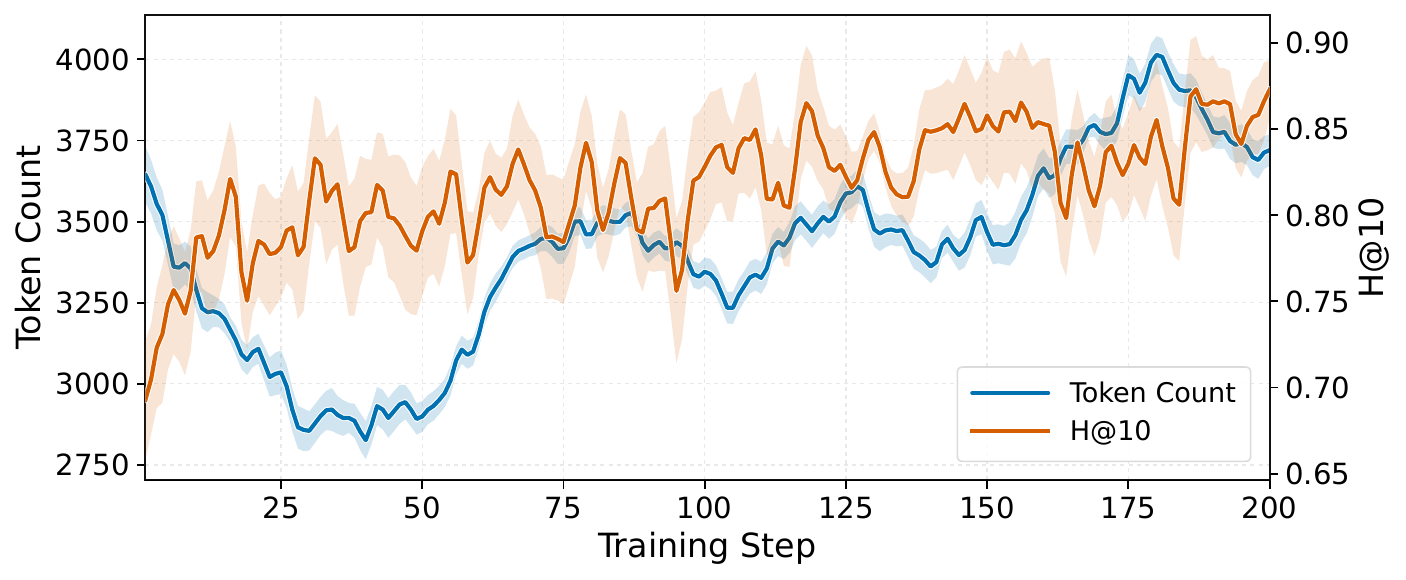}
    \caption{Dynamics of token count and average H@10 of rollout trajectories during first-stage training on the Games dataset.}
    \label{fig:token_vs_h10}
\end{figure}

\begin{figure*}[t]
    \centering
    \includegraphics[width=1\linewidth]{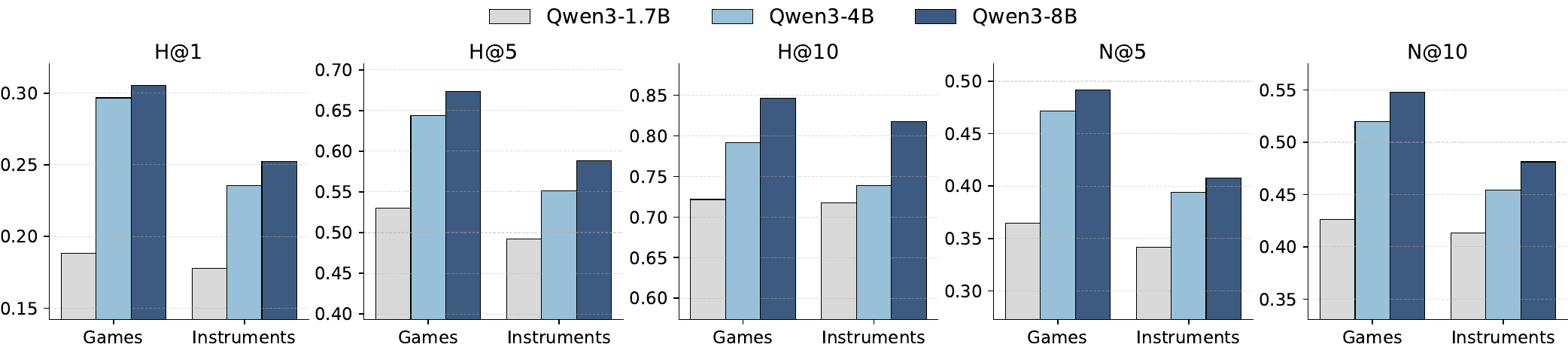}
    \caption{Performance of \ours with varying backbone sizes.}
    \label{fig:exp-model}
\end{figure*}

\begin{figure*}[t]
    \centering
    \includegraphics[width=1\linewidth]{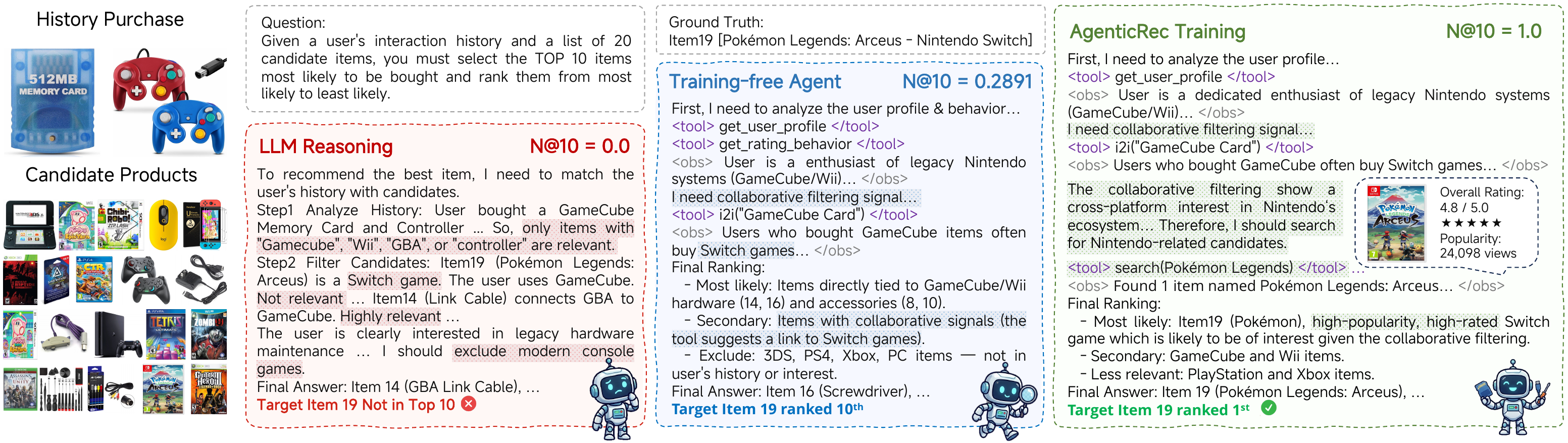}
    \caption{An example showing the effectiveness of \ours. Textured highlights denote critical reasoning steps that most significantly influence the final ranking and illustrate the differences between methods.}
    \label{fig:casestudy}
\end{figure*}

We conduct an ablation study to examine the effectiveness of the two stage training paradigm and the overall advantage of the agentic framework. As shown in Tab.~\ref{tab:ablation}:
\begin{itemize}[leftmargin=10pt,topsep=3pt,itemsep=0.2pt]
    \item Training consistently improves the performance of both R and TIRR. This verifies that optimizing the model with recommendation-specific signals can effectively enhance recommendation quality, while the frozen remain suboptimal.

    \item The two-stage training paradigm is effective. Compared with RTA, RTA+PPR further improves the final recommendation performance, showing that PPR brings real gains on top of RTA. By replaying self-bootstrapped ranking errors and optimizing the agent with bidirectional preference reasoning, PPR refines the agent's preference boundary and transfers this ability to final recommendation.

    \item The agentic setting mostly achieves better performance than R, validating the advantage of tool-integrated reasoning recommendation. However, TIRR underperforms R on Instruments and Office in the frozen setting, which is expected: without training, tool invocation relies only on language priors and may introduce noisy evidence through suboptimal tool calls. After training, TIRR consistently outperforms R across all datasets, highlighting the necessity of our training method.
\end{itemize}

\subsection{Further Analysis}
\label{sec:further_analysis}

\paragraph{Effect of PPR Training}
To further examine the effect of PPR, we track the agent's performance during PPR training in Fig.~\ref{fig:ppr_training}.
During PPR training, the agent's preference judgment accuracy on the PPR test set continues to increase.
Meanwhile, the agent's H@1 on the final recommendation test set also steadily improves.
This verifies that PPR can enhance the model's ability to distinguish fine-grained preferences.
It also suggests that the agent implicitly generalizes a more versatile ability, leading to more accurate recommendations.

\paragraph{Token Count}
Fig.~\ref{fig:token_vs_h10} illustrates the dynamics of token count and average H@10 of rollout trajectories during first-stage training on the Games dataset.
During training, the agent is encouraged to produce longer reasoning chains and conduct deeper exploration, leading to recommendations that better satisfy user preferences.
Meanwhile, the dialogue token cost remains within an acceptable range. More statistics about token cost are in Appendix~\ref{app:addition_exp}.

\paragraph{Latency}
We report the average latency per recommendation request on a single NVIDIA A800 GPU using vLLM. As shown in Table~\ref{tab:agent_latency}, the latency remains practical for an agentic recommendation framework. Unlike conventional recommendation pipelines optimized for strict real-time ranking, \ours targets interactive and reasoning-intensive recommendation scenarios, where richer reasoning and evidence-grounded decision making are often preferred over millisecond-level response latency.

\begin{table}[t]
    \centering
    \small
    \setlength{\tabcolsep}{5pt}
    \renewcommand{\arraystretch}{0.9}
    \begin{tabular}{@{}lrr@{}}
    \toprule
    Data & Avg. Latency (s) & Avg. Tool Time (s) \\
    \midrule
    CDs & 16.233 & 0.010 \\
    Games & 18.260 & 0.022 \\
    Office & 20.678 & 0.019 \\
    Instruments & 17.074 & 0.013 \\
    \bottomrule
    \end{tabular}
    \caption{Average latency statistics for our method.}
    \label{tab:agent_latency}
\end{table}

\paragraph{Scaling Backbone LLM}
To examine the scalability of \ours with respect to model capacity, we conduct experiments using Qwen3~\cite{Qwen3} across varying sizes (1.7B, 4B, and 8B). 
As shown in Fig.~\ref{fig:exp-model}, performance improves consistently as model size increases across datasets and evaluation metrics, indicating that \ours can effectively utilize increased model capacity to enhance performance.

\subsection{Case Study}
\label{sec:case_study}

To intuitively demonstrate the effectiveness of \ours, we present a sample from the Games dataset in Fig.~\ref{fig:casestudy}. The user history consists of legacy Nintendo purchases (GameCube accessories), while the actual next item purchased is \textit{Pokémon Legends: Arceus}, a game for the Nintendo Switch. The vanilla reasoning LLM exhibits an over-reliance on literal string matching (e.g., strictly filtering for keywords like \textit{GameCube}), failing to mine latent interests or grasp the intrinsic semantics of items. By invoking tools, the training-free agent utilizes collaborative signals to improve ranking quality; however, without RTA training and PPR refinement, it fails to further explore nuanced differences among items and remains trapped by literal semantics in the final decision, resulting in suboptimal performance.

In contrast, \ours leverages collaborative signals to astutely infer the user's latent interest in the Nintendo ecosystem. Building on this insight, it proactively invokes tools to acquire detailed specifications of promising candidates, enabling fine-grained differentiation for ranking. Ultimately, it flexibly synthesizes multi-source information to deliver precise recommendations.
\section{Conclusion}
\label{sec:conclusion}

In this paper, we proposed \ours, which casts agentic recommendation as tool-integrated reasoning workflow and optimizes it via a two-stage paradigm: Recommendation-Oriented Trajectory Activation activates agentic recommendation ability, while Progressive Preference Refinement sharpens preference discrimination. Experiments confirm the effectiveness of \ours.
\section*{Limitations}



Limited by computational resources, we did not conduct experiments with larger-scale LLM backbones, such as 72B-level models. Our preliminary experiments show that \ours benefits from increased backbone capacity, and we expect stronger abilities to emerge when \ours is instantiated with larger LLMs. In addition, this work focuses on optimizing recommender agents from offline interaction logs. A promising future direction is to further explore self-evolving recommender agents to improve the agent's adaptability in dynamic recommendation environments.

\bibliography{ref}

\newpage
\appendix
\section{Prompts Used in \ours}
\label{app:prompt}

Following are the prompts used in \ours:
\definecolor{blue}{RGB}{34, 139, 230}
\begin{tcolorbox}[
    colframe = gray,       
    colback = gray!5!white,             
    coltitle = white,                   
    coltext = black,                    
    fonttitle = \bfseries,              
    title = System Prompt for \ours,  
    boxrule = 1pt,                      
    arc = 2mm,                          
    width = \linewidth,                 
    left = 7pt,                         
    right = 7pt,                        
    top = 5pt,                          
    bottom = 5pt                        
]
\fontsize{8.5pt}{10pt}\selectfont
You are a helpful recommendation system assistant equipped with external tools. You can optionally call the tools to assist with the task. \par

\vspace{8pt}
\textbf{Tools} \par
\vspace{3pt}

You may call one or more functions to assist with the user query. You are provided with function signatures within \texttt{<tools></tools>} XML tags: \par

\vspace{5pt}
\texttt{<tools>} \par
\textcolor{blue}{\texttt{\{tools\_schema\}}} \par
\texttt{</tools>} \par

\vspace{8pt}
For each function call, return a json object with function name and arguments within \texttt{<tool\_call></tool\_call>} XML tags: \par

\vspace{5pt}
\texttt{<tool\_call>} \par
\texttt{\{} \par
\hspace{1.5em} \texttt{"name": <function-name>,} \par
\hspace{1.5em} \texttt{"arguments": <args-json-object>} \par
\texttt{\}} \par
\texttt{</tool\_call>}
\end{tcolorbox}

\begin{tcolorbox}[
    colframe = gray,       
    colback = gray!5!white,             
    coltitle = white,                   
    coltext = black,                    
    fonttitle = \bfseries,              
    title = Prompt for Recommendation,  
    boxrule = 1pt,                      
    arc = 2mm,                          
    width = \linewidth,                 
    left = 7pt,                         
    right = 7pt,                        
    top = 5pt,                          
    bottom = 5pt                        
]
\fontsize{8.5pt}{10pt}\selectfont
Solve the following problem step by step. You now have the ability to selectively invoke tools to gather information to assist your reasoning. \par
The last part of your response should be in the following format: \par
\vspace{3pt}
<answer> \par
\texttt{\textbackslash boxed\{'The final answer goes here.'\}}\par
</answer>

\vspace{5pt}
\textbf{User question:} \par
Given a user's interaction history and a list of 20 candidate items, you must select the TOP 10 items most likely to be bought and rank them from most likely to least likely.\par

\vspace{5pt}
\textbf{Answer Format:} \par
1. Output ONLY the \texttt{\textbackslash boxed\{[...]\}} result with no additional commentary in your final answer. \par
\quad - Example: \texttt{\textbackslash boxed\{[3, 17, 8, 1, 12, 5, 19, 2, 14, 6]\}} \par
2. Your final output MUST include exactly 10 item indices (no more, no less). \par
3. Each index must appear exactly once (no duplicates). \par
4. All indices must be valid (between 1 and 20). \par

\vspace{8pt}
Here is the history of items the user has bought: \par
\textcolor{blue}{\texttt{\{history\_items\_str\}}}

\vspace{8pt}
Here is the list of candidate items: \par
\textcolor{blue}{\texttt{\{candidate\_items\_str\}}}

\vspace{8pt}
Remember to place the final answer in the last part using the format: \par
<answer> \par
\texttt{\textbackslash boxed\{'The final answer goes here.'\}} \par
</answer>
\end{tcolorbox}

\begin{tcolorbox}[
    colframe = gray,       
    colback = gray!5!white,             
    coltitle = white,                   
    coltext = black,                    
    fonttitle = \bfseries,              
    title = Prompt for PPR Training,  
    boxrule = 1pt,                      
    arc = 2mm,                          
    width = \linewidth,                 
    left = 7pt,                         
    right = 7pt,                        
    top = 5pt,                          
    bottom = 5pt                        
]
\fontsize{8.5pt}{10pt}\selectfont
Solve the following problem step by step. You need to selectively invoke tools to gather information to assist your reasoning. \par
The last part of your response should be in the following format: \par

\vspace{3pt}
<answer>\par
\texttt{\textbackslash boxed\{'The final answer goes here.'\}}\par
</answer>

\vspace{8pt}
\textbf{User question:} \par
Given a user's interaction history and two candidate items, you must determine which item the user is \textcolor{blue}{\textbf{[MORE / LESS]}} LIKELY to buy next. \par

\vspace{5pt}
\textbf{Answer Format:} \par
1. Output ONLY the \texttt{\textbackslash boxed\{A\}} or \texttt{\textbackslash boxed\{B\}} result in your final answer. \par
2. Your answer must be exactly one of: A or B. \par

\vspace{8pt}
Here is the history of items the user has bought: \par
\textcolor{blue}{\texttt{\{history\_items\_str\}}}

\vspace{8pt}
Here are the two options: \par
Item A: \textcolor{blue}{\texttt{\{item\_a\_info\}}} \par
Item B: \textcolor{blue}{\texttt{\{item\_b\_info\}}}

\vspace{8pt}
Remember to place the final answer in the last part using the format: \par
<answer> \par
\texttt{\textbackslash boxed\{'The final answer goes here.'\}} \par
</answer>
\end{tcolorbox}

\section{Theoretical Analysis}
\label{app:theoretical_analysis}

\subsection{Directional Unbiasedness of RTA Gradient}
\label{app:unbias}

The RTA gradient estimator theoretically remains directionally unbiased up to a positive scaling factor under the group-average baseline:
\begin{proposition}[Directional Unbiasedness of RTA Gradient]
\label{prop:unbias}
    The RTA gradient estimator, which utilizes the list-wise ranking metric $R(r_K, y)$ as the reward and the inclusive group average ranking score as the baseline, provides a directionally unbiased estimate of the gradient for the expected list-wise utility objective $J(\theta) = \mathbb{E}_{\tau \sim \pi_\theta}[R(r_K, y)]$, up to the positive scaling factor $\frac{G-1}{G}$.
\end{proposition} 
Proposition~\ref{prop:unbias} confirms that the computed gradient direction aligns mathematically with maximizing the expected list-wise utility (i.e., NDCG@K). The inclusive baseline introduces only a positive scaling factor $\frac{G-1}{G}$ when $G>1$, which acts as a learning rate dampener without altering the optimization direction. This directional unbiasedness is critical for \ours for two reasons: (1) First, it ensures accurate credit assignment for tool usage, guaranteeing that the ``credit'' from a high-quality ranking is correctly back-propagated to the specific intermediate reasoning and tool-invocation steps, thereby driving the agent to learn outcome-driven behaviors. (2) Second, it provides stability in sparse feedback environments; by anchoring the gradient update to a group baseline, the optimization focuses on relative list improvement rather than unstable absolute scores, which is essential for convergence when learning from noisy implicit ranking signals.

\begin{proof}
In \ours, the optimization goal is to maximize the expected quality of the generated top-$K$ list $r_K$ given user context $x_u$. The trajectory $\tau$ encompasses the entire reasoning chain, including tool invocations (Act) and observations (Obs), culminating in the final ranking action $r_K$. Thus, the objective is:
\begin{equation}
\nabla J(\theta) = \mathbb{E}_{\tau \sim \pi_\theta} [R(r_K, y) \cdot \nabla \log \pi_\theta(\tau)]
\end{equation}
where $R(r_K, y)$ is the list-wise metric (i.e., NDCG@K) derived from implicit feedback $y$ (Eq.~\ref{eq:rankreward}).

To address the high variance inherent in list-wise ranking rewards, we employ the GRPO estimator using a group of $G$ sampled trajectories $\{\tau^{(g)}\}_{g=1}^G$. The estimated gradient is:
\begin{equation}
\begin{aligned}
        \nabla \hat{J}_{GRPO}(\theta) &= \frac{1}{G} \sum_{g=1}^{G} \\
        &(R(r_K^{(g)}, y) - b) \nabla \log \pi_\theta(\tau^{(g)})
\end{aligned}
\end{equation}
where the baseline $b = \frac{1}{G} \sum_{j=1}^G R(r_K^{(j)}, y)$ is the average list-wise utility of the sampled group.

With the inclusive group-average baseline, the baseline term does not vanish exactly because $b$ contains the reward of the same trajectory $\tau^{(g)}$. Let $R^{(g)}=R(r_K^{(g)}, y)$ and $z^{(g)}=\nabla \log \pi_\theta(\tau^{(g)})$. Using the likelihood ratio identity, we have:
\begin{equation}
    \label{eq:baseline}
    \resizebox{\linewidth}{!}{$
    \begin{aligned}
        \mathbb{E}_{\tau} [\nabla \log \pi_\theta(\tau)] &= \int \pi_\theta(\tau) \cdot \nabla \log \pi_\theta(\tau) \, d\tau \\
        &= \int \pi_\theta(\tau) \cdot \frac{\nabla \pi_\theta(\tau)}{\pi_\theta(\tau)} \, d\tau \\
        &= \int \nabla \pi_\theta(\tau) \, d\tau \\
        &= \nabla \left( \int \pi_\theta(\tau) \, d\tau \right) \\
        &= \nabla (1) = 0,
    \end{aligned}
    $}
\end{equation}
where $\int \pi_\theta(\tau) \, d\tau = 1$ since the integral of any probability distribution over its entire domain must sum to 1.

Therefore, for each $g$,
\begin{equation}
\resizebox{\linewidth}{!}{$
\begin{aligned}
\mathbb{E}[b \cdot z^{(g)}]
&= \mathbb{E}\left[\frac{1}{G}\sum_{j=1}^{G} R^{(j)} z^{(g)}\right] \\
&= \frac{1}{G}\mathbb{E}[R^{(g)} z^{(g)}] + \frac{1}{G}\sum_{j\neq g}\mathbb{E}[R^{(j)}]\mathbb{E}[z^{(g)}] \\
&= \frac{1}{G}\nabla J(\theta).
\end{aligned}
$}
\end{equation}
It follows that
\begin{equation}
\resizebox{\linewidth}{!}{$
\begin{aligned}
\mathbb{E}[\nabla \hat{J}_{GRPO}(\theta)]
&= \frac{1}{G}\sum_{g=1}^{G}\left(\mathbb{E}[R^{(g)}z^{(g)}]-\mathbb{E}[b z^{(g)}]\right) \\
&= \left(1-\frac{1}{G}\right)\nabla J(\theta) \\
&= \frac{G-1}{G}\nabla J(\theta).
\end{aligned}
$}
\end{equation}
Since $G>1$, $\frac{G-1}{G}$ is strictly positive. Thus, the gradient direction remains aligned with the true objective of maximizing NDCG@K, while the group-relative subtraction $(R^{(g)} - b)$ reduces variance by focusing on the relative ranking quality within the group.
\end{proof}

\subsection{Error Bound Minimization via Bidirectional Preference Reasoning}
\label{app:bidirectional}

The bidirectional preference reasoning in PPR stage can theoretically sharpen the fine-grained preference boundaries among highly confusable candidates:  
\begin{proposition}[Error Bound Minimization via Bidirectional Preference Reasoning]
\label{prop:bidirectional}
    Optimizing the bidirectional preference reasoning objective on mined hard negative pairs $(c^+, c^-)$ minimizes the upper bound of the pairwise ranking error probability $P\left(rank_{r_K}(c^-) < rank_{r_K}(c^+)\right)$.
\end{proposition} 
Proposition~\ref{prop:bidirectional} shows that, by explicitly training on ranking violations (where initially $\Delta s = s(c^+) - s(c^-) < 0$) and applying the bidirectional update, the agent maximizes the score margin $\Delta s$. The ``Negative Direction'' task provides an additional gradient flow explicitly penalizing the salient features of $c^-$ that caused the violation, thereby tightening the error bound more effectively than positive-only supervision alone. 

\begin{proof}
Let $s(c | x_u)$ be the agent's scoring potential for an item $c$ given context $x_u$ and we use $s(c)$ for short if there is no ambiguity. 
A ranking violation occurs if $s(c^-) > s(c^+)$ for a hard negative $c^-$. We define the score difference as:
\begin{equation}
    \Delta s = s(c^+) - s(c^-).
\end{equation}
A ranking error occurs when $\Delta s < 0$.

The fundamental goal of ranking optimization is to minimize the misordering of pairs. The discrete pairwise ranking error is given by the indicator function:
\begin{equation}
    \mathcal{L}_{0-1}(\Delta s) = \mathbb{I}(\Delta s < 0) = \begin{cases} 1 & \text{if } \Delta s < 0 \\ 0 & \text{if } \Delta s \ge 0 \end{cases}
\end{equation}
Directly optimizing $\mathcal{L}_{0-1}$ is intractable due to its discontinuity and zero gradients almost everywhere, which prevents effective gradient-based updates via the agent's policy.

For a pair $(c^+, c^-)$, our bidirectional preference reasoning creates two complementary tasks: 
\begin{itemize}[leftmargin=10pt,topsep=3pt,itemsep=0.2pt]
    \item \textbf{Positive Direction:} Maximize likelihood of choosing $c^+$. Let this probability be: 
    $$P_{pos} = \frac{e^{s(c^+)}}{e^{s(c^+)} + e^{s(c^-)}}$$
    
    \item \textbf{Negative Direction:} Maximize likelihood of rejecting $c^-$ (identifying it as ``less likely''). This effectively maximizes: $$P_{neg} = \frac{e^{-s(c^-)}}{e^{-s(c^+)} + e^{-s(c^-)}} = \frac{e^{s(c^+)}}{e^{s(c^+)} + e^{s(c^-)}}$$
\end{itemize}

Optimizing these tasks via GRPO is equivalent to minimizing the negative log-likelihood (cross-entropy) of the preference probability $P(c^+ \succ c^-)$. 
For the pair $(c^+, c^-)$, the combined loss $\mathcal{L}_{Bi}$ is proportional to:
\begin{equation}
\begin{aligned}
    \mathcal{L}_{Bi} &\approx - \log(P_{pos}) - \log(P_{neg}) \\
    &= - 2 \log \left( \frac{1}{1 + e^{-(s(c^+) - s(c^-))}} \right) \\
    &= 2\log(1 + e^{-\Delta s}).
\end{aligned}
\end{equation}
The above formulation represents the logistic loss.

We observe that the derived logistic loss function $\mathcal{L}(\Delta s) = 2\log(1 + e^{-\Delta s})$ serves as a convex upper bound to the 0-1 ranking error loss $\mathcal{L}_{0-1}$:
\begin{itemize}[leftmargin=10pt,topsep=5pt,itemsep=0.2pt]
    \item \textbf{Upper Bound:} For any $\Delta s$, $2\log(1 + e^{-\Delta s}) \ge 2\log(2) \cdot \mathbb{I}(\Delta s < 0)$. Specifically, when $\Delta s < 0$ (an error), the logistic loss grows linearly with the magnitude of the violation ($-\Delta s$), imposing a heavy penalty. When $\Delta s > 0$ (correct ranking), the loss approaches zero smoothly.

    \item \textbf{Convexity:} Unlike the step function $\mathbb{I}(\cdot)$, the logistic loss $\mathcal{L}(\Delta s)$ is smooth and convex. This ensures that the optimization landscape is well-behaved, providing non-vanishing gradients even when the ranking is currently incorrect (i.e., when $\Delta s \ll 0$).
\end{itemize}

By mining hard negatives where $\Delta s < 0$  and applying the bidirectional update, \ours effectively minimizes $2\log(1 + e^{-\Delta s})$. Since this function is a tight convex upper bound to the indicator error function, minimizing the bidirectional loss theoretically guarantees the minimization of the empirical pairwise ranking error rate, thereby sharpening the fine-grained preference boundaries among highly confusable candidates.
\end{proof}
\section{Experimental Details}
\label{app:experimental_details}

\subsection{Datasets and Preprocess}
\label{app:datasets_and_preprocess}

We conduct experiments on 4 subsets of the Amazon Reviews 2023 corpus~\cite{00020W0Z0LH025}, including CDs and Vinyl (CDs), Musical Instruments (Instruments), Office Products (Office), and Video Games (Games). To keep the training time tractable, we restrict the interaction data to the period from October 2022 to October 2023. We adopt a common chronological split with a ratio of 8:1:1 for training/validation/testing by timestamp, where interactions are sorted by time and split accordingly. The statistics of these processed datasets are summarized in Tab.~\ref{tab:datasets}.

For interaction behavior sequence construction, we follow \citet{llara} to cap the maximum length of user interaction history to 10.
For candidate construction, each instance includes one ground-truth next item and 19 randomly sampled negatives, forming 20 candidates in total. When sampling negatives, we exclude all items appearing in the user's historical interaction sequence to avoid overlap with previously interacted items. We then randomly shuffle the candidate items before feeding them to the model, so that the input order does not introduce positional bias.

\begin{table}[!ht]
\centering
\caption{Statistics of the preprocessed datasets.}
\resizebox{\linewidth}{!}{
    \begin{tabular}{lrrrr}
    \toprule
    \textbf{Dataset} & \textbf{\#Users} & \textbf{\#Items} & \textbf{\#Inters.} & \textbf{Sparsity} \\
    \midrule
    CDs & 5,437 & 8,785 & 13,826 & 99.71\% \\
    Instruments & 7,593 & 5,279 & 15,746 & 99.61\% \\
    Office & 27,130  & 11,511  & 47,333  & 99.85\% \\
    Games & 6,251 & 3,003 & 11,457 & 99.39\% \\
    \bottomrule
    \end{tabular}
}
\label{tab:datasets}
\end{table}

\subsection{Baselines}
\label{app:baselines}

\paragraph{Conventional Sequential Recommenders} include classic baselines Caser~\cite{Caser}, GRU4Rec~\cite{GRU4Rec}, and SASRec~\cite{SASRec}, which model user preferences from interaction sequences using convolutional, recurrent, and self-attention architectures, respectively. We further consider ReaRec~\cite{ReaRec}, an inference-enhanced recommender that introduces additional reasoning steps at test time.

\paragraph{Training-free LLM-based Recommenders} mainly rely on prompting or in-context reasoning without parameter updates. In particular, LLMRank~\cite{LLMRank} directly performs candidate ranking by prompting an off-the-shelf LLM with user history and item descriptions, serving as a strong zero-shot ranking baseline. InteRecAgent~\cite{InteRecAgent} represents autonomous recommender agents that leverage an LLM to interact with external tools (i.e., SASRec) for recommendation decision making, while remaining training-free for the LLM backbone.

\paragraph{Trainable LLM-based Recommenders} adapt LLMs to recommendation via fine-tuning or reinforcement learning. TALLRec~\cite{TALLRec} improves recommendation by tuning LLMs with recommendation-oriented supervision under limited data. LLaRA~\cite{llara} further integrates textual and ID-style representations to align LLMs with recommendation signals. S-DPO~\cite{S-DPO} introduces a multi-negative, softmax-style DPO objective to better leverage implicit preference data beyond positive-only supervision. ReRe~\cite{ReRe} applies constrained beam search and preference rewards to improve recommendation result.

\subsection{Implementation Details}
\label{app:implementations_details}

For \ours, We use Qwen3-4B-Instruct-2507~\cite{Qwen3} as the backbone. The training follows the two-stage training with 3 and 1 epochs, respectively, and we set the group size to 8, the maximum generation length to 6,144, the batch size to 64, the learning rate to $1\times10^{-6}$, and the rollout temperature to 0.6.

For sequential recommendation baselines, we use a batch size of 1,024 and optimize all models with binary cross-entropy loss and a learning rate of 0.001. 
For ReaRec, the number of reasoning steps during training is set to 3.
For training-free LLM-based baselines, we use GPT-4 as the LLM for all experiments. 
For trainable LLM-based baselines, to ensure a fair comparison, we adopt Qwen3-4B-Instruct-2507 as the backbone across all methods, the same as \ours, with a learning rate of $1\times10^{-5}$ and a batch size of 128, while keeping other hyperparameters and configurations as their default settings. 
Since TALLRec~\cite{TALLRec} was originally designed for point-wise ranking, we adapt it to our setting by reformulating the task as selecting the next item from a given candidate set. Similarly, ReRe~\cite{ReRe} is adapted to our setting by constraining the model to recommend from a given candidate set.
For baselines that explicitly output item relevance scores, we obtain the ranked list by scoring all candidate items and sorting them in descending order. For TALLRec, LLaRA, S-DPO and ReRe, we rank the candidate items according to their joint generation probability. 

All experiments are conducted on 4 NVIDIA A800 GPUs (80GB memory each).
\section{Additional Empirical Results}
\label{app:addition_exp}

\paragraph{Analysis on Tool Invocation}
Fig.~\ref{fig:training_dynamics}(a) shows the statistics of tool invocation during the first stage training (RTA) on Office. The tool invocation rate among positively rewarded trajectories (orange line) increases rapidly in early training and remains consistently high thereafter, indicating that the policy quickly learns to leverage tools for effective decision-making. Meanwhile, the average number of tool calls (blue line) gradually rises and then stabilizes, suggesting that the model develops a relatively stable tool use strategy rather than indiscriminately increasing tool calls.

Fig.~\ref{fig:training_dynamics}(b) shows the change of H@10 on training trajectories during the first stage training, reflecting steady improvements in the learned policy. The synchronized progression between stabilized tool use and increasing H@10 suggests that our training method brings more effective tool planning, gradually leading to better recommendation.

\begin{figure}[t]
  \centering
  \begin{subfigure}[b]{0.50\linewidth} %
    \includegraphics[width=\linewidth]{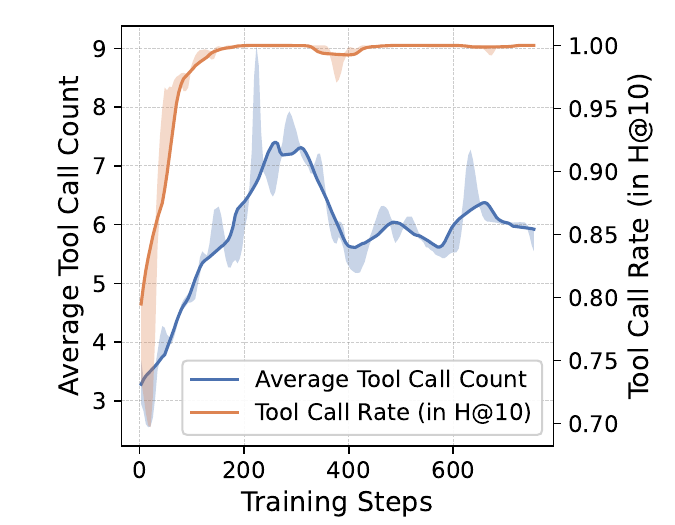} %
    \caption{Tool use} %
    \label{fig:tool-usage}
  \end{subfigure}%
  \hfill %
  \begin{subfigure}[b]{0.50\linewidth}
    \includegraphics[width=\linewidth]{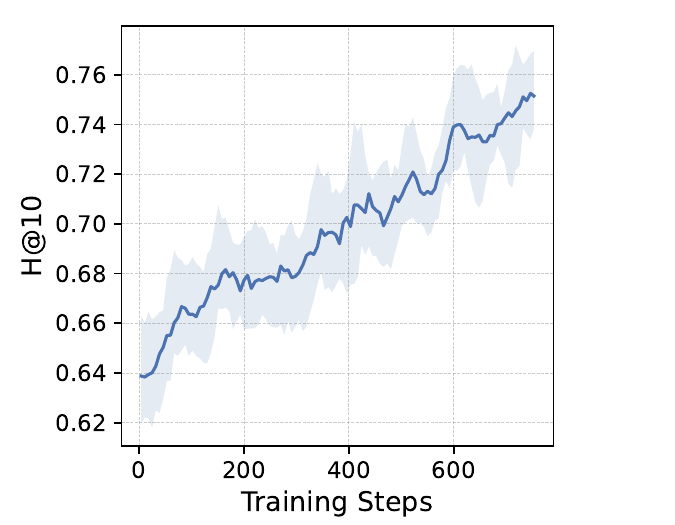} %
    \caption{Recommendation quality} %
    \label{fig:tool-score}
  \end{subfigure}
  \vspace{-20pt}
  \caption{Training statistics on the Office dataset: (a) tool usage statistics, where the blue line denotes the average number of tool calls per trajectory and the orange line indicates the percentage of positively rewarded trajectories that have tool invocation; (b) recommendation performance measured by H@10 during training.}
  \vspace{-5pt}
  \label{fig:training_dynamics}
\end{figure}

\paragraph{Analysis of Group Size}
As shown in Fig.~\ref{fig:exp-group}, increasing the group size in RTA improves recommendation performance by providing richer intra-group comparisons for more stable list-wise credit assignment. However, the gains gradually saturate with larger groups. Thus, a moderate group size offers a practical trade-off between training stability and efficiency.

\begin{figure}[t]
    \centering
    \includegraphics[width=1.0\linewidth]{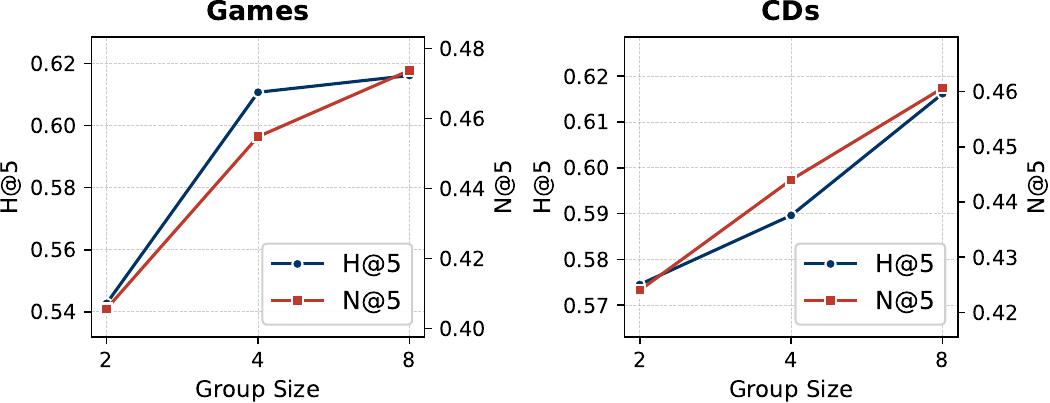}
    \vspace{-20pt}
    \caption{Effect of group size in RTA}
    \vspace{-10pt}
    \label{fig:exp-group}
\end{figure}

\paragraph{Token Cost}
Fig.~\ref{fig:token_statistics} reports the token cost statistics of \ours across datasets.
Naive LLM denotes the token consumption when directly using the LLM to reason and generate recommendation results.
\ours denotes the token count during recommendation generation after training under the tool-integrated reasoning recommendation paradigm, including \textit{Think}, \textit{Act}, \textit{Obs}, and \textit{Rec} tokens. Compared with standard LLM reasoning, the token cost of \ours remains relatively controllable.

\begin{figure}[t]
    \centering
    \includegraphics[width=1\linewidth]{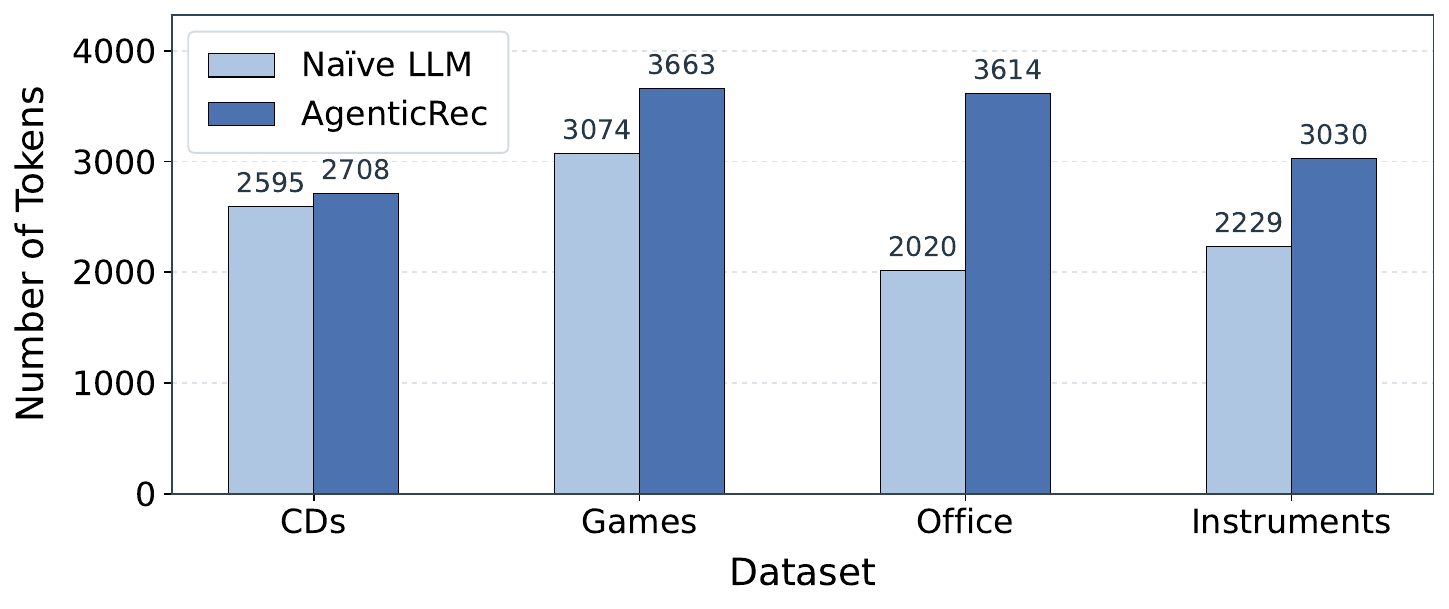}
    \caption{Full token cost statistics.}
    \label{fig:token_statistics}
\end{figure}

\section{Tool Example}
\label{app:tool}

The following are examples of tool invocation in \ours:
\begin{itemize}[leftmargin=10pt,topsep=5pt,itemsep=0.2pt]
    \item \textbf{User Profile Tool:} Fig. \ref{fig:user_profile_tool} illustrates the user profile tool. We utilize the Qwen3-4B-Instruct-2507 (consistent with our backbone) to pre-generate user profiles. The input consists of the user's interaction history prior to the timestamp of each sample in our processed dataset.
    
    \item \textbf{Item Information Tool:} As shown in Fig. \ref{fig:item_information_tool}, \texttt{item\_info\_search} tool retrieves detailed information about a specific item based on the dataset's item metadata. Additionally, we construct a candidate analysis function to support another tool \texttt{candidates\_analyze} that aggregates candidate items by category or price, providing a structured summary that enables the agent to quickly grasp the composition of the candidate list. This dual-granularity design facilitates high-quality list-wise ranking by combining specific item details with a candidate overview.
    
    \item \textbf{Behavioral Statistics Tool:} Fig. \ref{fig:behavioral_statistics_tool} illustrates the return information of this tool. Upon invocation, it computes statistics on relevant metadata fields to generate statistical summaries of user recent behavior. Specifically, for \textit{Session Behavior} (\texttt{get\_session\_behavior}), the user's dynamic short-term interest patterns are extracted using item metadata from recent sessions; meanwhile, for \textit{Rating Behavior} (\texttt{get\_rating\_behavior}), items rated by the user are categorized into high, neutral, and low groups to serve as explicit preference signals. 
    
    \item \textbf{Collaborative Information Tool:} Fig. \ref{fig:collaborative_information_tool} illustrates the information returned by the tool. We employ SASRec\cite{SASRec} on the training set to model collaborative information and construct the embedding space, which is utilized for collaborative retrieval.
\end{itemize}

\definecolor{MainFrameGray}{HTML}{808080} 
\definecolor{HeaderGray}{HTML}{E9F1F6}    
\definecolor{ToolBlue}{HTML}{2563EB}      
\definecolor{OutputGray}{HTML}{64748B}    
\definecolor{JSONKey}{HTML}{BE185D}       
\definecolor{CodeBg}{HTML}{f9f9f9}        
\lstdefinelanguage{modernjson}{
    basicstyle=\small\ttfamily\color{black},
    columns=fullflexible,
    breaklines=true,
    showstringspaces=false,
    stringstyle=\color{black!80},
    literate=
     *{:}{{{\color{ToolBlue}{:}}}}{1}
      {"}{{{\color{JSONKey}{" \kern-1.5pt}}}}{1} 
      {\{}{{{\color{black}{\{}}}}{1}
      {\}}{{\color{black}{\}}}}{1},
}
\tcbset{
    solidstyle/.style={
        enhanced,
        boxrule=1pt,           
        arc=2pt,               
        left=2mm, right=2mm, top=1mm, bottom=1mm,
        fonttitle=\small\bfseries\sffamily,
        shadow={0pt}{0pt}{0pt}{white} 
    }
}
\newtcblisting{toolcall}{
    solidstyle,
    colback=white,
    colframe=MainFrameGray,
    coltitle=white,
    colbacktitle=MainFrameGray,
    title={$\blacktriangleright$ Tool Call},
    listing only,
    listing options={
        language=modernjson,
        basicstyle=\small\ttfamily\linespread{0.8}\selectfont, 
        aboveskip=0pt, 
        belowskip=0pt  
    }
}
\newtcolorbox{tooloutput}{
    solidstyle,
    colback=white,
    colframe=MainFrameGray, 
    coltitle=MainFrameGray,
    colbacktitle=HeaderGray,
    title={$\blacktriangleleft$ Tool Output},
    fontupper=\small\rmfamily\color{black},
}
\newtcolorbox{interactionwrapper}[1]{
    enhanced,
    colback=CodeBg,
    colframe=MainFrameGray,
    boxrule=1.0pt,             
    arc=4pt,
    title={#1},
    fonttitle=\small\bfseries\sffamily\color{white},
    colbacktitle=MainFrameGray,
    attach boxed title to top left={xshift=-1.2pt, yshift=-2mm}, 
    boxed title style={
        arc=2pt,
        outer arc=2pt,
        boxrule=1.0pt, 
        colframe=MainFrameGray, 
        colback=MainFrameGray
    },
    top=4mm,
    bottom=2mm,
    left=2mm,
    right=2mm,
    before skip=10pt,
    after skip=10pt
}

\begin{figure}[!ht]
\centering
\begin{interactionwrapper}{User Profile Tool}
\begin{toolcall}
<tool_call>
{
  "name": "get_user_profile",
  "arguments": {}
}
</tool_call>
\end{toolcall}
\begin{tooloutput}
The user is a dedicated enthusiast of Nintendo legacy systems, particularly favoring Switch controllers and related accessories. They consistently prefer high-functionality products with features like vibration, motion control, and macro capabilities.
\end{tooloutput}
\end{interactionwrapper}
\caption{An example of user profile tool invocation.} 
\label{fig:user_profile_tool}
\end{figure}

\begin{figure}[!ht]
\centering
\begin{interactionwrapper}{Item Information Tool}
\begin{toolcall}
<tool_call>
{
    "name": "item_info_search",
    "arguments": {
        "item_name": "Plants vs. Zombies Garden Warfare 2 - Xbox One"
    }
}
</tool_call>
\end{toolcall}
\begin{tooloutput}
Found 1 item named "Plants vs. Zombies Garden Warfare 2 - Xbox One" (Categories: Xbox One, Games / Store: Electronic Arts / Price: \$16.99) \par
Overall Rating: 4.6 / 5.0 \par
Popularity: 1625 reviews \par
\end{tooloutput}
\begin{toolcall}
<tool_call>
{
    "name": "candidates_analyze",
    "arguments": {}
}
</tool_call>
\end{toolcall}
\begin{tooloutput}
Candidate Analysis by Category: \par
- Legacy Systems > Nintendo Systems > Game Boy Systems > Game Boy: \par
\qquad -- Item 13: [EverDrive GBA Mini] \par
- Legacy Systems > Nintendo Systems > Nintendo DS > Games: \par
\qquad -- Item 4: [Call of Duty: Black Ops] \par
\qquad -- Item 17: [Cake Mania NDS] \par
- Legacy Systems > PlayStation Systems > PlayStation 2 > Accessories: \par
\qquad -- Item 6: [PS2 Lava Glow 2.4 GHZ RF Wirel...] \par
\qquad -- Item 12: [Wireless 2.4G Controller Gamep...] \par
- Nintendo Switch > Games: \par
\qquad ...
\end{tooloutput}
\end{interactionwrapper}
\caption{An example of item information tool.} 
\label{fig:item_information_tool}
\end{figure}

\begin{figure}[!ht]
\centering
\begin{interactionwrapper}{Behavioral Statistics Tool}
\begin{toolcall}
<tool_call>
{
    "name": "get_session_behavior",
    "arguments": {}
}
</tool_call>
\end{toolcall}
\begin{tooloutput}
Below are the extracted interests for the user's most recent 2 sessions, listed chronologically from oldest to newest: \par
3.3h ago, the user was interested in Games, Legacy Systems. \par
1.8h ago, the user was interested in PlayStation 4, Games. \par

This information reflects the user's historical interests, interest evolution, or potential current focus.
\end{tooloutput}
\begin{toolcall}
<tool_call>
{
  "name": "get_rating_behavior",
  "arguments": {}
}
</tool_call>
\end{toolcall}
\begin{tooloutput}
User's historical rating records are as follows: \par
Five-star items (=5): \par
\qquad PowerA MOGA XP5-X Plus Bluetooth Controller for Mobile \& Cloud Gaming on Android/PC... \par
Neutral items (>=3): \par
\qquad sunwaytek Syozen Z2 Wireless Mechanical Controller for Nintendo Switch... \par
\qquad Wireless Switch Pro Controller with 9 Color Adjustable LED/Unique Crack/Turbo/Vibration/Motion Control/Wake-Up... \par
\qquad Dishonored: The Death of the Outsider... \par
Low-rated items (<3): \par
\qquad PowerA Enhanced - Seafoam Fade, Gamepad, Wired Video Game Controller, Gaming Controller, Xbox Series X|S, Xbox One (Only at Amazon)...
\end{tooloutput}
\end{interactionwrapper}
\caption{An example of behavioral statistics tool.} 
\label{fig:behavioral_statistics_tool}
\end{figure}

\begin{figure}[!ht]
\centering
\begin{interactionwrapper}{Collaborative Information Tool}
\begin{toolcall}
<tool_call>
{
    "name": "get_similar_items",
    "arguments": {
    "item_title": "Tiger Woods PGA Tour"
    }
}
</tool_call>
\end{toolcall}
\begin{tooloutput}
Based on collaborative filtering, users who enjoy the following items are more likely to purchase [Tiger Woods PGA Tour]: \par
[1] Nintendo Switch Mario Kart 8 Deluxe Wheel \par
[2] GIRIAITUS Wii Motion Plus Adapter \par
[3] Resident Evil 5: Gold Edition...
\end{tooloutput}
\begin{toolcall}
<tool_call>
{
    "name": "get_similar_users",
    "arguments": {}
}
</tool_call>
\end{toolcall}
\begin{tooloutput}
The following users have similar preferences to the current user. Their recent purchases include: \par
User [1]: [Madden NFL 10 - Playstation 3], [Dirt - Xbox 360], [NBA Hangtime - PlayStation]...
\end{tooloutput}   
\end{interactionwrapper}
\caption{An example of collaborative information tool invocation.} 
\label{fig:collaborative_information_tool}
\end{figure}

\end{document}